\documentclass[10pt,journal,letterpaper]{IEEEtran}
\IEEEoverridecommandlockouts
\usepackage{graphicx}
\usepackage{amsmath}
\usepackage{amsfonts}
\usepackage{amssymb}
\usepackage{cite}
\usepackage{color}
\usepackage{multirow}
\usepackage{enumerate}
\usepackage{tabularx}
\usepackage{amsthm}

\usepackage{booktabs}
\usepackage{tabularx}
\usepackage{tabularx}

\usepackage{algorithmicx}
\usepackage{algorithm}
\usepackage{algpseudocode}

\newcolumntype{L}[1]{>{\raggedright\arraybackslash}p{#1}}
\newcolumntype{C}[1]{>{\centering\arraybackslash}p{#1}}
\newcolumntype{R}[1]{>{\raggedleft\arraybackslash}p{#1}}

\usepackage{times, epsfig}
\usepackage{subfig}

\makeatletter
\newcommand{\multiline}[1]{%
	\begin{tabularx}{\dimexpr\linewidth-\ALG@thistlm}[t]{@{}X@{}}
		#1
	\end{tabularx}
}
\makeatother

\newtheorem{Lem}{Lemma}

\newlength{\figwidth}
\begin{document}
	
	\setlength{\pdfpagewidth}{8.5in}
	\setlength{\pdfpageheight}{11in}
\title{\LARGE Relay Assisted Underlay Cognitive Radio Networks with Multiple Users}

\author{Lanwei Zhang,~\IEEEmembership{Student Member,~IEEE,}
        Rajitha Senanayake,~\IEEEmembership{Member,~IEEE,}
        \\Saman Atapattu,~\IEEEmembership{Senior Member,~IEEE,}
        and~Jamie Evans,~\IEEEmembership{Senior Member,~IEEE.}
	\vspace{0mm} }
\maketitle

\begin{abstract}
In this letter,  we consider an underlay cognitive radio network assisted by dual-hop decode-and-forward (DF) relaying. For a general multi-user network, we adopt a \textit{max-min fairness} relay selection scheme and analyse the outage probability when the channels are subject to independent and non-identical Nakagami-$m$ fading. The relay network  operates  within  the  constraint  imposed  on  the peak interference power tolerable by the primary receiver. We then analyse the asymptotic outage probability performance and illustrate the existence of i) the full-diversity order when the interference level at the primary user increases proportionally with the relay transmit power; and ii) an outage floor when the transmit powers of the relays are restricted by the primary receiver. We also analyse the outage probability with imperfect channel state information (CSI) and the average throughput over Rayleigh fading channels. 
Illustrative analytical results are accurately validated by numerical simulations.
\end{abstract}

\vspace{-0mm}
\section{Introduction}\label{introduction}
Cognitive radio is proposed as a promising technique to mitigate the inefficient utilization of the radio frequency spectrum~\cite{Ali17}.
Underlay cognitive radios allow unlicensed secondary users to access the licensed frequency spectrum, under the condition that they do not cause an intolerable level of interference to the licensed users \cite{Haykin05}. 
Relaying has emerged as a technique  to improve power efficiencies, coverage and reduce outage probability. It has been extensively studied in the context of cognitive radio networks, e.g., \cite{Hyadi13,Park16,Boddapati18,Sun20,Duong12,Hyadi132,Lei2017tccn}.
\vspace{-0mm}
\subsection{Related Work}
Several research efforts explored the performance in terms of outage probability, bit error rate, secrecy outage and/or throughput  of relay assisted cognitive radio networks \cite{Hyadi13,Park16,Boddapati18,Lei2017tccn}. New relay selection strategies have also been proposed  under maximum transmit power and maximum interference constraints imposed by the primary network \cite{Sun20}. All these works consider a single secondary-user pair, and the performance has been analyzed over Rayleigh fading channels. Urban areas, however, are candidates for network densification due to the high concentration of mobile users. As a result, the development of a relay assisted cognitive radio network to support multiple secondary-user communications within a given coherence time is of paramount importance, especially in 4G, 5G and beyond wireless applications in disaster areas, battlefield or tactical operations. Few papers have considered  the more practical Nakagami fading channel model \cite{Duong12,Hyadi132,Lei2017tccn} but only for simple single secondary-user networks. 

\vspace{-0mm}
\subsection{Our Contribution}
Since previous papers have not considered a {\it joint relay selection} for {\it multiple secondary-user networks} over {\it Nakagami} fading, this paper focuses on a multi-user cognitive radio network where multiple secondary users transmit information to multiple secondary destinations via  decode-and-forward (DF) relays. Due to the power constraints imposed by the primary network, the transmit powers of the relays are restricted to the peak interference power tolerable by the primary receiver. For such a network, we make the following novel contributions: i) We first adopt a {\it max-min fairness} RS scheme which was proposed for a classical multi-user multi-relay dual-hop network in \cite{Atapattu13}.  This optimal RS (ORS) scheme maintains user fairness while maximizing the received signal-to-noise ratio (SNR); ii) We then analytically characterize the outage probability for the ORS scheme when the channels are subject to  Nakagami-$m$ fading and distance dependent path-loss; iii) Asymptotic analysis shows that the ORS scheme can achieve a full-diversity order under certain SNR conditions, or else there exists  an outage floor due to the power constraints of the primary network; and iv) We extend the analysis to consider imperfect channel state information (CSI) over  Rayleigh fading, as a special case. v) We analyse the average throughput over Rayleigh fading channels. The above results are also compared with those for other RS schemes, e.g. Naive RS scheme and Random RS scheme. Moreover, the considered  dual-hop network and the corresponding analysis can  be extended to a multi-hop network by adopting a decentralized RS algorithm proposed in \cite{Senanayake18,Sun20}. 

\vspace{-0mm}
\section{System Model}\label{systemmodel}\label{SM}
\vspace{-0mm}
\subsection{Network Model}
We consider an underlay cognitive radio network where $M$ secondary transmitters ($S_1, \ldots S_{M}$)  communicate with $M$ dedicated secondary receivers ($D_1, \ldots D_{M}$), in the presence of a primary receiver ($PU$). The communication is assisted by a dual-hop DF relay network with $N$ relay nodes ($R_1, \ldots R_{N}$). We have $M$ secondary node pairs denoted as user~$u$, $u\in\{1,2, \ldots M\}$, where user~$u$ corresponds to the $S_u-D_u$ pair.
We assume that there are no direct links between secondary transmitters and secondary receivers and each user pair is assisted by only one relay which cannot be shared by more than one user. Thus, we need $N\ge M$ relays in the network.
Since all $M$ user pairs in the secondary network communicate in a given coherence time without interfering with each other, we divide each coherence time into $2M$ time-slots and implement orthogonal transmission which is easy to implement with less feedback overhead, compared to other transmission schemes such as non-orthogonal multiple access (NOMA). We make the common assumption that the interference caused to the secondary receivers by the primary  transmission is negligible~\cite{Hyadi13}. 
Similar to \cite{Nguyen20}, we also assume that the $PU$ is located closer to the second hop and as the secondary transmitters and relay nodes are geographically dislocated,  the interference caused to the $PU$ by the secondary transmitters is negligible when compared to the interference caused by the relay transmission.
Each node is equipped with a single antenna, and operates in half-duplex mode.
Without loss of generality, we consider that user~$i$ communicates via relay~$j$. For the first-hop, the channel coefficient and distance between $S_{i}$ and $R_{j}$ are denoted as $h_{ij,1}$ and $d_{ij,1}$, respectively. For the second-hop the the same parameter are denoted as $h_{ij,2}$ and $d_{ij,2}$. The distance-dependent path-loss is considered with path-loss exponent $\beta$.
For the interference channel, the channel coefficient and distance between $R_{j}$ and $PU$ are denoted as $f_{ij}$ and $d_{ij,3}$, respectively.

\vspace{-0mm}
\subsection{Analytical Model}\label{ss_anyalytical}
For the first hop, we denote $P$ as the fixed power budget of each source node.
For the second hop, in order to satisfy the interference and maximum power constraints, the  transmission power of $R_j$ is constrained as
$Q_{ij} = \textrm{min}\left[Q_{max},I_{max}/(|f_{ij}|^2/d_{3}^\beta)\right]$, where $I_{max}$ and $Q_{max}$ denote the peak interference power at $PU$ and  the peak transmission power of $R_j,~\forall j$, respectively.
Then, the end-to-end received SNR of user~$i$ helped by the DF relay $R_j$ is given by
\begin{align}\label{eq_gammaij}
    \gamma_{ij}=\min\left({\gamma}_{ij,1},{\gamma}_{ij,2} \right)
\end{align}
where
${\gamma}_{ij,1}={P}|h_{ij,1}|^2/(d_{1}^\beta  N_0)$ is the first-hop received SNR,  ${\gamma}_{ij,2}=Q_{ij}|h_{ij,2}|^2/(d_{2}^\beta  N_0)$
is the second-hop received SNR, and $N_0$ is the noise variance.
Thus, all possible user SNRs connected via any relay can be given in matrix form as
\begin{equation}\label{snrmatrix}
   {\bf\Gamma} = \left(\gamma_{ij}\right) \in \mathbb{R}^{M \times N}.
\end{equation}

Moreover, each channel amplitude within the same hop is assumed to be independent and identically distributed (i.i.d.) Nakagami-$m$ fading, while the channels from different hop can be non-identical. As such, the squared amplitude, $|h|^2\in\{|h_{ij,1}|^2,|h_{ij,2}|^2,|f_{ij}|^2\}$, follows a gamma distribution, denoted as $|h|^2\sim \mathcal{G}(m,\Omega)$, where $m\in \mathbb{Z}^+$ is   the shape parameter,  $\Omega=\mathbb{E}[|h|^2]$ and $\mathbb{E}[\cdot]$ represents the statistical expectation. Its cumulative distribution function (CDF) can then be given as
$F_{|h|^2}(x) = \gamma\left(m,mx/\Omega\right)/\Gamma(m)$
where $\Gamma(\cdot)$ is the gamma function and $\gamma(\cdot, \cdot)$ is the lower incomplete gamma function~\cite{TISP07}. Therefore, we have $|h_{ij,1}|^2\sim \mathcal{G}(m,\Omega_{h_1})$, $|h_{ij,2}|^2\sim \mathcal{G}(m,\Omega_{h_2})$ and $|f_{ij}|^2\sim \mathcal{G}(m,\Omega_{f})$.
Since we assume  that  nodes  in  each  hop  are  co-located or in a close neighbourhood, it is reasonable to take $d_{ij,1}=d_1$, $d_{ij,2}=d_2$ and $d_{ij,3}=d_3$ $\forall i,j$ where $d_1$, $d_2$ and $d_3$ can be distinct values.
\vspace{-0mm}
\subsection{Relay Selection (RS) Scheme}\label{ss_rs}
We consider the {\it max-min fairness} RS algorithm which guarantees the individual performance as well as user fairness, capable of choosing the set of paths that maximizes the minimum end-to-end SNR of all users~\cite{Atapattu13}. The algorithm can be performed on entries of the SNR matrix $\bf\Gamma$ in \eqref{snrmatrix}, and its output gives the optimal relay assignment for $M$ user pairs where we denote the RS matrix as $\hat{\bf\Gamma}$. Since a relay cannot be shared by more than one user, while each row of $\hat{\bf\Gamma}$ has only one entry, each column of $\hat{\bf\Gamma}$ has at most one entry.
Then, the effective SNR for a user is the SNR which effects the performance  of that user after applying the relay selection algorithm. We denote the effective SNR of user~$i$ as $\gamma_{(i)}$ and the $k$th largest entry in $\bf\Gamma$  as $\gamma^{(k)}$.
 Further, 
 $\gamma_{(i)}$ satisfies the property:  $\gamma_{(i)}\in\{\gamma^{(1)},\ldots,\gamma^{((M-1)N+1)}\}$
 . Thus,  $\gamma_{(i)}=\gamma^{((M-1)N+1)}$ represents the worst case which occurs when all  $\gamma^{((M-1)N+1)}\geq \cdots \geq\gamma^{(MN)}$ are in a same row  if $N>M$; or are either in same row or column if $N=M$. 
 More details of the algorithm can be found in~\cite{Atapattu13}.

\vspace{-0mm}
\section{Performance Analysis}\label{perf}
\vspace{-0mm}
\subsection{Outage Probability}\label{s_out}
The outage probability of user~$u$, $P_{o}^{(u)}$, is the probability that the SNR $\gamma_{(u)}$ falls below a certain predetermined threshold SNR $\gamma_{th}$. It can then be calculated  as $P_{o}^{(u)} = {\mathbb P}[\gamma_{(u)} \leq \gamma_{th}]$, and is given in the following lemma.
\begin{Lem}\label{Lem_user_outage}
For an underlay cognitive radio network with $M$ users and $N$ relays, the outage probability of each user with Max-Min fairness RS scheme  can be given as
\begin{equation}
\label{eq_user_outage}
\begin{split}
\hspace{-2mm}P_{o}^{(u)} =\hspace{-4mm}  \sum_{k = 1}^{(M-1)N +1}   \sum_{i=0}^{k-1}\frac{{\mathbb P}(\gamma^{(k)})(MN)!\binom{k-1}{i}[F_{\gamma_{ij}}(\gamma_{th})]^{\mu(k,i)}}{(-1)^i\mu(k,i) (k-1)!(MN-k)!}
\end{split}
\end{equation}
where $\mu(k,i)=M N - k + i + 1$,
\begin{equation*}
\label{cf_cdf_snr}
\begin{split}
     F_{\gamma_{ij}}(x) \hspace{-1mm}=& \frac{\gamma(m,\frac{mx}{\Omega_1\Lambda_1})}{\Gamma(m)} \hspace{-1mm}+\hspace{-1mm}
     \frac{\Gamma(m,\frac{mx}{\Omega_1\Lambda_1})\gamma(m,\frac{mx}{\Omega_2\Lambda_2}) \gamma(m,\frac{m\Lambda_3}{\Omega_3\Lambda_2})}{\Gamma(m)^3}\\ &-\frac{\Gamma(m,\frac{mx}{\Omega_1\Lambda_1})}{\Gamma(m)^2} \sum_{k=0}^{m-1}\hspace{-1mm} \frac{(\Omega_3x)^{k}\Gamma\hspace{-1mm}\left(\hspace{-1mm} k\hspace{-1mm}+\hspace{-1mm}m, \hspace{-1mm} \frac{m(\Omega_3x+\Omega_2\Lambda_3)}{\Omega_2\Omega_3\Lambda_2} \hspace{-1mm}\right)}{k!(\Omega_2\Lambda_3)^{-m} (\Omega_3x+\Omega_2\Lambda_3)^{k+m}}\\ &+\frac{\Gamma(m,\frac{mx}{\Omega_1\Lambda_1})}{\Gamma(m)^2} \Gamma(m,\frac{m\Lambda_3}{\Omega_3\Lambda_2})
\end{split}
\end{equation*}
with $\Omega_1 \hspace{-1mm}= \hspace{-1mm}{\Omega_{h_1}}/d_{1}^{\beta},\Omega_2 \hspace{-1mm}=\hspace{-1mm} {\Omega_{h_2}}/d_{2}^{\beta},\Omega_3\hspace{-1mm} = \hspace{-1mm}{\Omega_{f}}/d_{3}^{\beta}$, $\Lambda_1 \hspace{-1mm}= \hspace{-1mm}P/N_0, \Lambda_2\hspace{-1mm} =\hspace{-1mm} Q_{max}/N_0, \Lambda_3 \hspace{-1mm}= \hspace{-1mm}I_{max}/N_0$,  and $\Gamma(\cdot,\cdot)$ is the upper incomplete gamma function.
The term ${\mathbb P}(\gamma^{(k)})\hspace{-1mm}= \hspace{-1mm}{\mathbb P}\left[\gamma_{(u)} \hspace{-1mm}= \hspace{-1mm}\gamma^{(k)}\right]$ represents the probability that the relay with the $k$-th largest entry of  $\bf\Gamma$ corresponds to user~$u$.
\begin{proof}
	See Appendix~\ref{App_lem1}.
\end{proof}
\end{Lem}
For a general network, ${\mathbb P}\left[\gamma_{(u)} = \gamma^{(k)}\right]$ can be evaluated numerically by using simulations. For a two-user network, i.e, for $M=2$ and $N\geq 2$, ${\mathbb P}\left[\gamma_{(u)} = \gamma^{(k)}\right]$ can be calculated analytically based on  \cite[Sec. IV.C]{Senanayake18}, 
and the outage probability in \eqref{eq_user_outage} can then be given as a closed-form analytical expression.


\vspace{-0mm}
\subsection{Asymptotic Analysis}
In this section we provide the high SNR performance of the outage probability. In particular, we consider $P_{o}^{(u)}$ w.r.t. the SNR $\Lambda\rightarrow \infty$, where $\Lambda\in\{\Lambda_1,\Lambda_2,\Lambda_3 \}$, and provide analytical results for different scenarios in the following lemma.
\begin{Lem}\label{Lem_user_outage_asy}
For an underlay cognitive radio network with $M$ users and $N$ relays, asymptotic behaviours of each user's outage probability w.r.t. SNR $\Lambda$  with Max-Min fairness RS scheme  are given for two cases in the following.

\noindent
{{\it {\bf Case~1} (Diversity order and array gain)}}: When $\Lambda_1=\Lambda_2=\Lambda_3=\Lambda\rightarrow \infty$, we have
\begin{equation}
\label{eq_user_outage_asy1}
\begin{split}
P_{o}^{(u)} \rightarrow  \mathcal{A}\,\Lambda^{-mN} + \mathcal{O}\left( \Lambda^{-(mN+1)}\right)
\end{split}
\end{equation}
where each user achieves diversity order $mN$ with array gain $\mathcal{A}$ given by
\begin{equation}\label{array_gain}
\mathcal{A} = \left\{\begin{array}{ll}
        \frac{G(m)^N (MN)! \gamma_{th}^{mN} \prod_{i=1}^{N-1} \frac{N-i}{MN-i}}{M (MN-N)! N!}, & M > N \\
        \frac{2G(m)^{N} (MN)! \gamma_{th}^{mN} \prod_{i=1}^{N-1} \frac{N-i}{MN-i}}{M (MN-N)! N!}, & M = N
    \end{array},
    \right.
\end{equation}
with $G(m) = \frac{m^{m-1}}{\Gamma(m)\Omega_1^m} + \frac{m^m \gamma(m,\frac{m}{\Omega_3})+\Omega_3^m\Gamma(2m,\frac{m}{\Omega_3})}{m\Gamma(m)^2 \Omega_2^m}.$

\noindent{{\it {\bf Case~2} (Outage floor)}}: When $\Lambda_2=\Lambda\rightarrow \infty$ with fixed $\Lambda_1$ and $\Lambda_3$, we have
\begin{equation}
\label{eq_user_outage_asy2}
\begin{split}
\hspace{-3mm}P_{o}^{(u)} \rightarrow\hspace{-4mm}  \sum_{k = 1}^{(M-1)N +1}   \sum_{i=0}^{k-1}\hspace{-1mm}\frac{{\mathbb P}(\gamma^{(k)})(MN)!\binom{k-1}{i}[F_{\gamma}(\gamma_{th})]^{\mu(k,i)}}{(-1)^i\mu(k,i) (k-1)!(MN-k)!}
\end{split}
\end{equation}
where
    $F_{\gamma} (x) \hspace{-1mm}=\hspace{-1mm} 1 - \frac{\Gamma(m,\frac{mx}{\Omega_1\Lambda_1})}{\Gamma(m)^2}  \sum_{k=0}^{m-1} \hspace{-1mm} \frac{(\Omega_2\Lambda_3)^m(\Omega_3x)^k}{k!(\Omega_3x+\Omega_2\Lambda_3)^{k+m}} \Gamma(k+m)$.
Since $P_{o}^{(u)}$ is not a function of $\Lambda$, we have an outage floor w.r.t. $\Lambda_2$ for fixed $\Lambda_1$ and $\Lambda_3$.

\end{Lem}

\begin{proof}
	See Appendix~\ref{App_lem2}.
\end{proof}
\noindent
{\it {\bf Remark} } Although we consider only two cases in Lemma~\ref{Lem_user_outage_asy}, similarly, we can analyze other cases.
For all other cases, i) $\Lambda_3=\Lambda\rightarrow \infty$ with fixed $\Lambda_1$ and $\Lambda_2$; ii) $\Lambda_1=\Lambda\rightarrow \infty$ with fixed $\Lambda_2$ and $\Lambda_3$; iii) $\Lambda_1=\Lambda_2=\Lambda\rightarrow \infty$ with a fixed $\Lambda_3$;
iv) $\Lambda_2=\Lambda_3=\Lambda\rightarrow \infty$ with a fixed $\Lambda_1$;
and v) $\Lambda_1=\Lambda_3=\Lambda\rightarrow \infty$ with a fixed $\Lambda_2$, each user achieves an outage floor.

\vspace{-0mm}
\subsection{Imperfect CSI over Rayleigh Fading}
For brevity, we assume Rayleigh fading channels. Further, we assume that a master node which can be any user or relay estimates CSI requirements by using minimum mean square error (MMSE) estimation. Since the focus of our discussion is on the channel estimation errors, we use the MMSE  model discussed in \cite[Sec.~II]{wang12}.
We consider that true channels $h_{ij,1}$,$h_{ij,2}$ and $f_{ij}$ are estimated as $\hat h_{ij,1}$, $\hat h_{ij,2}$ and $\hat f_{ij}$, respectively, with Gaussian estimation errors $e_{h_{ij,1}}$,$e_{h_{ij,2}}$ and $e_{f_{ij}}$, respectively. Then we have $h_{ij,1} =\hat{h}_{ij,1} + e_{h_{ij,1}}$, $h_{ij,2} = \hat{h}_{ij,2} + e_{h_{ij,2}}$, $f_{ij}= \hat{f}_{ij} +e_{f_{ij}}$. Estimated channel coefficients and corresponding estimation errors are assumed to be independent. Due to complex Gaussian true channels, Gaussian estimation errors and estimated channels are distributed as $\hat{h}_{ij,1}\sim \mathcal{CN}(0,\Omega_{\hat h_1}),\hat{h}_{ij,2}\sim \mathcal{CN}(0,\Omega_{\hat h_2}),\hat{f}_{ij}\sim \mathcal{CN}(0,\Omega_{\hat f})$, $e_{h_{ij,1}}\hspace{-1mm}\sim\mathcal{CN}(0,\Omega_{e_1}), e_{h_{ij,2}}\hspace{-1mm}\sim \mathcal{CN}(0,\Omega_{e_2}), e_{f_{ij}}\hspace{-1mm}\sim \mathcal{CN}(0,\Omega_{e_3})$,
where $\Omega_{e_1}=\Omega_{h_1}-\Omega_{\hat h_1}$, $\Omega_{e_2}=\Omega_{h_2}-\Omega_{\hat h_2}$ and $\Omega_{e_3}=\Omega_{f}-\Omega_{\hat f}$.
Following Section~\ref{ss_anyalytical}, the   transmission   power   of $R_j$ is decided based on estimated channels as   $\hat Q_{ij} = \textrm{min}\left[Q_{max},\frac{I_{max}}{|\hat f_{ij}|^2/d_{3}^\beta }\right]$. For imperfect CSI case, the end-to-end received SNR of user~$i$ helped by the DF relay $R_j$ can also be given  as \eqref{eq_gammaij}
with ${\gamma}_{ij,1}\hspace{-1mm}=\hspace{-1mm}\frac{P}{P\Omega_{e_1}+d_{1}^\beta N_0}|\hat h_{ij,1}|^2$ and
${\gamma}_{ij,2}\hspace{-1mm}=\hspace{-1mm}\frac{\hat Q_{ij}}{\hat Q_{ij}\Omega_{e_1}+d_{2}^\beta N_0}|\hat h_{ij,2}|^2$.

\begin{Lem}\label{Lem_user_outage_imCSI}
For an underlay cognitive radio network with $M$ users and $N$ relays, the outage probability of each user with Max-Min fairness RS scheme under Rayleigh fading with imperfect CSI can also be given as \eqref{eq_user_outage} with
\begin{equation}\label{eq_Fgammaij_impcsi}
    F_{\gamma_{ij}} (x) \hspace{-1mm}= \hspace{-1mm}1 - e^{\hspace{-1mm}-x\left(\frac{\Omega_{e_1}+\frac{d_1^\beta}{\Lambda_1}}{\Omega_{\hat h_1}}+\frac{\Omega_{e_2}+\frac{d_2^\beta}{\Lambda_2}}{\Omega_{\hat h_2}}
\right)}\hspace{-1mm} \left( \hspace{-1mm} 1 - \frac{e^{-\frac{d_3^\beta \Lambda_3}{\Omega_{\hat f}\Lambda_2}}}{1 \hspace{-1mm}+\hspace{-1mm} \frac{d_3^\beta\Omega_{\hat h_2}\Lambda_3}{d_2^\beta\Omega_{\hat f} x}} \hspace{-1mm} \right).
\end{equation}
\end{Lem}

\begin{proof}
	See Appendix~\ref{App_lem3}.
\end{proof}

\noindent
{\it {\bf Remark} } When $\Lambda_1=\Lambda_2=\Lambda_3=\Lambda\rightarrow \infty$, we have
\begin{equation}
\label{eq_user_outage_asy1_csi}
\begin{split}
\begin{split}
\hspace{0mm}P_{o}^{(u)} \hspace{-1mm}\rightarrow\hspace{-4mm}  \sum_{k = 1}^{(\hspace{-0.5mm}M\hspace{-0.5mm}-\hspace{-0.5mm}1\hspace{-0.5mm})\hspace{-0.5mm}N \hspace{-0.5mm}+\hspace{-0.5mm}1}   \sum_{i=0}^{k-1}\frac{{\mathbb P}(\hspace{-0.5mm}\gamma^{(k)}\hspace{-0.5mm})(\hspace{-0.5mm}M\hspace{-0.5mm}N\hspace{-0.5mm})!\binom{k\hspace{-0.5mm}-\hspace{-0.5mm}1}{i}\hspace{-1.5mm}\left(\hspace{-1.5mm}1 \hspace{-1mm}-\hspace{-1mm} e^{\hspace{-1mm}-\gamma_{\hspace{-0.3mm}t\hspace{-0.3mm}h}\hspace{-0.5mm}\left(\hspace{-1mm}\frac{\Omega_{e_{\hspace{-0.3mm}1}}}{\Omega_{\hat h_{\hspace{-0.3mm}1}}}+\frac{\Omega_{e_{\hspace{-0.2mm}2}}}{\Omega_{\hat h_{\hspace{-0.2mm}2}}}\hspace{-1mm}\right)}\hspace{-1.5mm}\right)\hspace{-1mm}^{\mu(k,i)}}{(-1)^i\mu(k,i) (k-1)!(MN-k)!},
\end{split}
\end{split}
\end{equation}
where each user achieves an outage floor and completely looses the diversity order under imperfect CSI.

\vspace{-0mm}
\subsection{Average throughput over Rayleigh Fading}
The throughput of a user, 
which reflects the channel capacity of the given user in half-duplex relay with orthogonal transmission, 
is given by $\tau_{(u)} = \frac{1}{2M}\log_2(1+\gamma_{(u)})$ bpcu, where $\gamma_{(u)}$ is the instantaneous SNR with Max-Min RS scheme. 
\begin{Lem}\label{Lem_user_throughput}
    For an underlay cognitive radio network with $M$ users and $N$ relays, the average throughput of each user with Max-Min fairness RS scheme under Rayleigh fading is given by
    \begin{equation}\label{eq_user_average_throughput}
        \bar \tau_{(u)} \hspace{-1mm}=\hspace{-4mm} \displaystyle\sum_{k=1}^{(M-1)N+1}\sum_{i=0}^{MN-k}\frac{(MN)!\binom{MN-k}{i}\displaystyle\sum_{j=0}^{t}\binom{t}{j}b^{t-j}c^j h(j)}{2M\ln 2(-1)^i t (k-1)!(MN-k)!}
    \end{equation}
    where $t = k+i$, $a = \frac{1}{\Omega_1\Lambda_1}+\frac{1}{\Omega_2\Lambda_2}$, $b = 1-e^{-\frac{\Lambda_3}{\Omega_3\Lambda_2}}$, $c = d (1-b)$, $d = \frac{\Omega_2\Lambda_3}{\Omega_3}$, $w = e^{dat} Ei(-dat)$, $Ei(x)$ is the exponential integral and 
    \begin{equation*}\label{eq_hj_1}
        h(j|d=1) \hspace{-1mm}= \hspace{-1mm}\left\{\hspace{-2mm}\begin{array}{lr}
            -e^{at} Ei(-at), & j=0 \\
            \frac{(-at)^j}{j!}\left(\sum_{l=1}^{j}\frac{(l-1)!}{(-at)^l} + h(0)\right), & j\geq 1 
        \end{array}\right.
    \end{equation*}
    \begin{equation*}\label{eq_hj_2}
        h(j|d\neq 1) \hspace{-1mm}= \hspace{-1mm}\left\{\hspace{-2mm}\begin{array}{lr}
            h(0|d=1), & j=0 \\
            \frac{h(0)+w}{d-1}, & j=1\\
            \frac{h(0) + w - \sum_{r=1}^{j-1} \hspace{-1mm} \frac{(-at)^r}{r!(d-1)^{-r}} \left[\sum_{l=1}^{r-1} \hspace{-1mm} \frac{(l-1)!}{(-dat)^l} - w \right]}{(d-1)^{j}}, & j\geq 2
        \end{array}\right.
    \end{equation*}
\end{Lem}
\begin{proof}
	See Appendix~\ref{App_lem4}.
\end{proof}

\vspace{-0mm}
\section{Numerical Examples}\label{numericals}
\begin{figure}[t]
    \centerline{\includegraphics[width=0.5\textwidth]{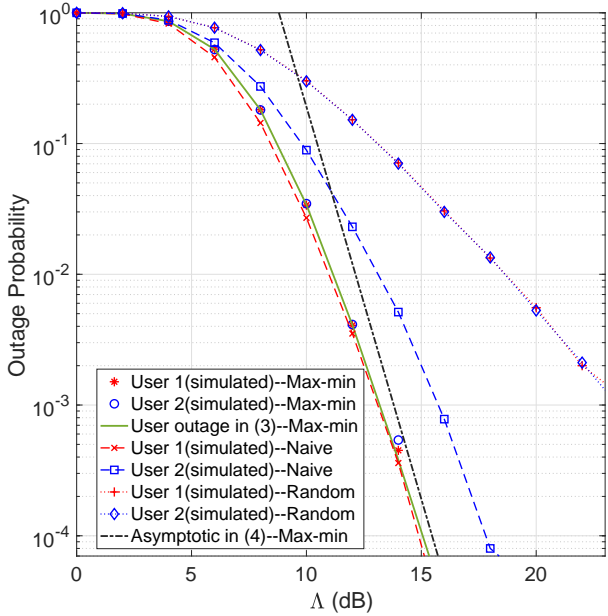}}
    \caption{Outage probability vs $\Lambda$ for a two-user three-relay network with $m=2$ and $\gamma_{th} = 5$\,dB}
    
    \vspace{-0mm}
    \label{compare}
\end{figure}

In this section, we present numerical examples to validate
our analysis and evaluate the performance of the Max-Min fairness RS scheme when compared to other common RS schemes. In these examples, we assume $\Omega_i = 1$, for $i=1,2,3$.


Fig.~\ref{compare} plots and compares the outage probability of each user in a two-user three-relay network over Nakagami-$2$ fading of the  Max-Min fairness  RS scheme with that of the naive RS (NRS) and random RS (RRS) schemes detailed in \cite{Atapattu13}. 
In order to observe the achievable diversity order we let $\Lambda_1 = \Lambda_2 = \Lambda_3 =  \Lambda$. The simulation results are generated using Monte-Carlo simulations while the exact and asymptotic  analytical results are generated using \eqref{eq_user_outage},  \eqref{eq_user_outage_asy1}, respectively. Confirming the accuracy of our analysis we observe that our analytical results match the simulation results for the entire simulated SNR range. Further, the asymptotic curves accurately quantify the outage probability performance in the high SNR regime, and we also observe a diversity order of $mN=6$. Under the Max-Min fairness  RS scheme, we observe that both users have the same outage probability, thus confirming the fairness among users. The User 1 of NRS just outperforms the Max-Min fairness RS scheme. However, the User 2 of NRS has a significantly worse performance with a diversity order of $m(N-1)=4$ which illustrates how user fairness is not maintained in NRS. Since RRS is equivalent to single relay case, it has the worst performance with a diversity order of $m=2$. 


Fig.~\ref{f_3N_g1} plots the outage probability of each user in a three-user three-relay network
under Nakagami-$m$ fading where $m = 1, 3$. We fix  $\Lambda_1$ and  $\Lambda_3$ and change $\Lambda_2$ to observe an outage floor. The exact and asymptotic  analytical results are generated using \eqref{eq_user_outage},  \eqref{eq_user_outage_asy2}, respectively. As expected we observe an outage floor which is accurately quantified by \eqref{eq_user_outage_asy2} in the high SNR regime.
\begin{figure}[t]
    \centerline{\includegraphics[width=0.5\textwidth]{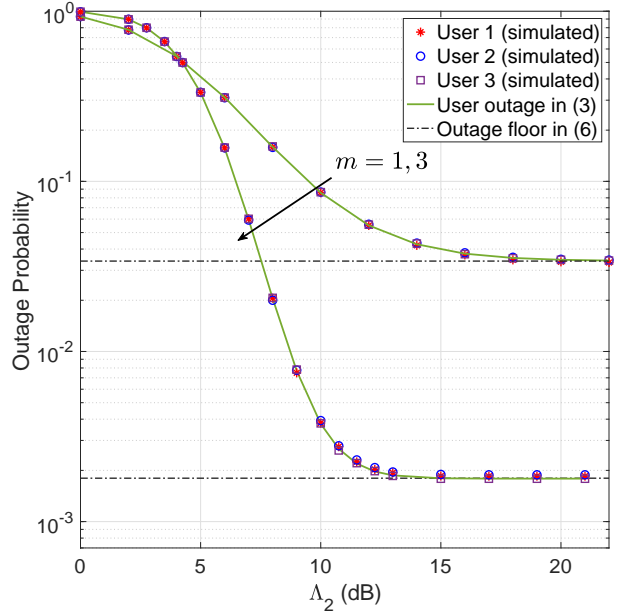}}
    \caption{Outage probability vs $\Lambda_2$ for a three-user three-relay network when $\Lambda_1=25~dB$, $\Lambda_3=10~dB$ and $\gamma_{th} = 5~dB$}
    \vspace{-0mm}
    \label{f_3N_g1}
\end{figure}

Fig.~\ref{imCSI} plots the user outage probability with imperfect CSI and Rayleigh fading, where $\Omega_{e_1}/\Omega_{h_1}\hspace{-1mm} =\hspace{-1mm}\Omega_{e_2}/\Omega_{h_2}\hspace{-1mm} =\hspace{-1mm}\Omega_{e_3}/\Omega_{f}\hspace{-1mm} = \hspace{-1mm}0.05$. We also set $\Lambda\hspace{-1mm}=\hspace{-1mm}\Lambda_1\hspace{-1mm}=\hspace{-1mm}\Lambda_2\hspace{-1mm}=\hspace{-1mm}\Lambda_3$ and consider two network conditions, namely a three-user four-relay network and a four-user four-relay network. The analytical results of imperfect CSI are generated from the expressions in Lemma~\ref{Lem_user_outage_imCSI}, which can be verified by the simulated results. We can observe that user outage will achieve full diversity order with perfect CSI, while in the imperfect CSI case,  it will reach a floor in the high SNR regime. This is reasonable in the sense that channel estimation error becomes the dominant error in the high SNR regime, and it will not change with SNR.
\begin{figure}[t]
    \centerline{\includegraphics[width=0.5\textwidth]{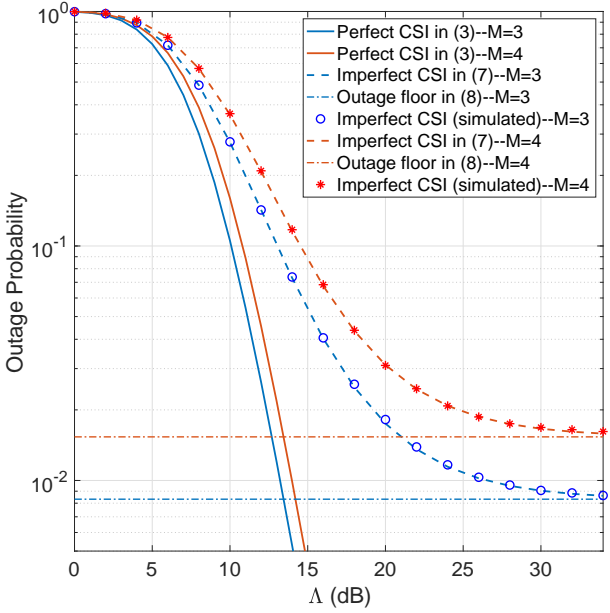}}
    \caption{Outage probability vs $\Lambda$ for a four-relay network for imperfect CSI with $M=3,4$ and $\gamma_{th} = 5$ dB}
    \vspace{-0mm}
    \label{imCSI}
\end{figure}


Fig.~\ref{tp_L2_3x4} plots the average throughput for a three-user four-relay network over Rayleigh fading channels. The simulation result for max-min RS scheme verified our analytical result in \eqref{eq_user_average_throughput}. The simulated average throughput of each user for two other RS schemes are also plotted. Naive RS scheme is not a fair scheme. It only guarantees the optimality of user 1 while sacrifices the performance of other users. Random RS scheme remains the worst performance.


\begin{figure}[t]
    \centerline{\includegraphics[width=0.5\textwidth]{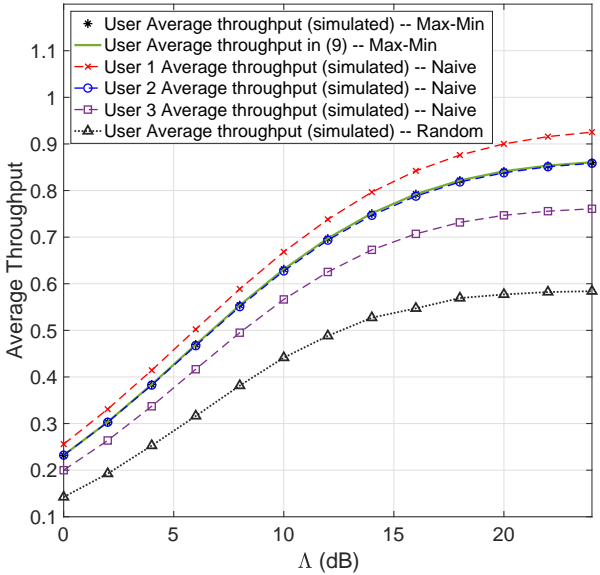}}
    \caption{Average throughput vs $\Lambda$ for a three user four-relay network with $\Lambda_2 = \Lambda, \Lambda_1=25~dB$, $\Lambda_3=10~dB$}
    \vspace{-0mm}
    \label{tp_L2_3x4}
\end{figure}

\vspace{-0mm}
\section{Conclusion}\label{conclusion}
The outage probability and average throughput  performance of an underlay cognitive radio network with multiple secondary users is analysed when the communication is assisted by dual-hop DF relaying. We investigate the performance of the Max-Min fairness RS scheme when the channels are subject to Nakagami-$m$ and/or Rayleigh fading channels. We conduct an asymptotic analysis which revealed that the Max-Min fairness RS scheme achieves the full diversity order of $mN$ when the interference level at the primary user increases proportionally with the relay power, and it has an outage floor when the primary user has a strictly fixed interference level.
Further, our RS scheme and analytical framework can easily be extended to a general multi-hop relay network where RS can be implemented hop-by-hop in a distributed manner\cite{Senanayake18,Sun20}. For a general $L$-hop network, we first develop individual SNR matrices $\bf \Gamma$s  for first $(L-2)$ hops and the joint SNR matrix $\bf \Gamma$ for the last two hops.  Then we apply the max-min fair RS algorithm hop-by-hop as explained in [11].

\begin{appendices}
\vspace{-0mm}
\section{Proof of Lemma~\ref{Lem_user_outage}}\label{App_lem1}
Based on the Max-Min RS scheme, $\gamma_{(u)}$ can take values $\gamma^{(1)}, \cdots, \gamma^{((M-1)N +1)}$.
Thus, we can write the outage probability of user~$u$ as 
\begin{equation}
\label{eq_outuseru}
\begin{split}
P_{o}^{(u)} 
= \sum_{k = 1}^{(M - 1)N +1} {\mathbb P}[\gamma_{(u)} = \gamma^{(k)}]F_{\gamma^{(k)}}(\gamma_{th}),
\end{split}
\end{equation}
where ${\mathbb P}[\gamma_{(u)} = \gamma^{(k)}]$ denotes the  probability $\gamma_{(u)} = \gamma^{(k)}$ and $F_{\gamma^{(k)}}(\gamma_{th})$ is the CDF of the $k$th largest element among $MN$ independent entries in ${\bf {\Gamma}}$ which can be written as \cite[eq. 8]{Atapattu13}
\begin{equation}\label{eq_orderstat}
    \hspace{-0mm}F_{\gamma^{(k)}}(x) =\hspace{-2mm} \sum_{i=0}^{k-1}\frac{(MN)!\binom{k-1}{i}(-1)^i[F_{\gamma_{ij}}(x)]^{MN-k+i+1}}{(MN-k+i+1)(k-1)!(MN-k)!}
\end{equation}
where $F_{\gamma_{ij}}(x)$ is the CDF of $\gamma_{ij}=\min\left({\gamma}_{ij,1},{\gamma}_{ij,2} \right)$. Since  ${\gamma}_{ij,1}$  and ${\gamma}_{ij,2}$ are mutually independent and Nakagami-$m$ distributed, the   $F_{\gamma_{ij}}(x)$ for a given $Q_{ji} = y$ can be derived as
\begin{align}
    \hspace{-2mm}F_{\gamma_{ij}|Q_{ji}} (x|y) \hspace{-1mm}
    &= \hspace{-1mm}\frac{\gamma(m,\frac{mx}{\Omega_1\Lambda_1})}{\Gamma(m)} + \frac{\Gamma(m,\frac{mx}{\Omega_1\Lambda_1})\gamma(m,\frac{mN_0x}{\Omega_2y})}{\Gamma(m)^2} \label{eq_Fgammaijy}
\end{align}
where $F_{{\gamma}_{ij,1}}(x)\hspace{-1mm} = \hspace{-1mm} \gamma(m,\frac{mx}{\Omega_1\Lambda_1})/\Gamma(m)$ with $\Omega_1 = {\Omega_{h_1}}/d_{1}^{\beta},\Lambda_1 = P/N_0$, and $F_{{\gamma}_{ij,2}|Q_{ij}}(x|y) = \gamma(m,\frac{mN_0x}{\Omega_2y})/\Gamma(m)$ with $\Omega_2 = {\Omega_{h_2}}/d_{2}^{\beta}$.
Next we note that the CDF of $Q_{ji}$ is given by
\begin{align}\label{eq_FQij}
    F_{Q_{ji}}(y) \hspace{-1mm}=\hspace{-1mm} \Big\{\begin{array}{ll}
        \hspace{-1mm}1-F_{|f_{ij}|^2} (d_3^\beta I_{max}/y), & \hspace{-1mm}0\leq y<Q_{max} \\
        \hspace{-1mm}1, & \hspace{-1mm}y\geq Q_{max}
    \end{array}.
\end{align}
We then evaluate its probability density function (PDF) as 
\begin{equation}\label{eq_fQij}
\begin{split}
    f_{Q_{ji}}(y) \hspace{-1mm}=\hspace{-1mm} \left(\hspace{-1mm}\frac{mI_{max}}{\Omega_3}\hspace{-1mm}\right)^{\hspace{-1mm}m} \hspace{-2mm}\frac{e^{-\frac{mI_{max}}{\Omega_3y}}U(y)}{\Gamma(m)y^{m+1}}
    \hspace{-1mm}+\hspace{-1mm}\frac{\gamma(m,\frac{m\Lambda_3}{\Omega_3\Lambda_2})D(y)}{\Gamma(m)},
\end{split}
\end{equation}
where $U(y)\hspace{-1mm}=\hspace{-1mm}\left(u(y)\hspace{-1mm}-\hspace{-1mm}u(y\hspace{-1mm}-\hspace{-1mm}Q_{max})\right)$, $D(y)\hspace{-1mm}=\hspace{-1mm}\delta\left(y\hspace{-1mm}-\hspace{-1mm}Q_{max}\right)$ with $u(\cdot)$ and $\delta(\cdot)$ denoting step and  impulse functions. The CDF of $\gamma_{ij}$ can then be derived with the aid of \eqref{eq_Fgammaijy} and \eqref{eq_fQij} as
\begin{equation}
\label{eq_Fgammaij_gx}
\begin{split}
     & F_{\gamma_{ij}}(x) =  \int F_{\gamma_{ij}|Q_{ji}} (x|y) f_{Q_{ji}}(y) dy\\
    & =  \frac{\gamma(m,\frac{mx}{\Omega_1\Lambda_1})}{\Gamma(m)} +\frac{\Gamma(m,\frac{mx}{\Omega_1\Lambda_1})\gamma(m,\frac{mx}{\Omega_2\Lambda_2}) \gamma(m, \frac{m\Lambda_3}{\Omega_3\Lambda_2})}{\Gamma(m)^3} \\
     &~~+ \frac{\Gamma(m,\frac{mx}{\Omega_1\Lambda_1})}{\Gamma(m)^2} \underbrace{\left(\frac{m\Lambda_3}{\Omega_3}\right)^{\hspace{-1mm} m}\hspace{-2mm}\int_{0}^{\Lambda_2}\frac{\gamma(m,\frac{mx}{\Omega_2z}) e^{-\frac{m\Lambda_3}{\Omega_3z}}}{\Gamma(m)z^{m+1}}dz}_{:= g(x)}
     ,
\end{split}
\end{equation}
where $\Omega_3\hspace{-1mm} = \hspace{-1mm}\Omega_f/d_3^\beta, \Lambda_2 \hspace{-1mm}= \hspace{-1mm}Q_{max}/N_0, \Lambda_3 \hspace{-1mm}= \hspace{-1mm}I_{max}/N_0$.
Further, the second equality comes by substituting \eqref{eq_Fgammaijy} and \eqref{eq_fQij}, and then by using term-by-term integration under variable transformation $z = y/N_0$. With aid of \cite[Eqs.~8.352.2 and 8.356.3]{TISP07}, we can solve $g(x)$ as $g(x) = \Gamma(m,\frac{m\Lambda_3}{\Omega_3\Lambda_2}) -  \sum_{k=0}^{m-1} \frac{(\Omega_2\Lambda_3)^m(\Omega_3x)^k \Gamma\left( k+m, \hspace{-1mm} \frac{m(\Omega_3x+\Omega_2\Lambda_3)}{\Omega_2\Omega_3\Lambda_2} \right) }{k!(\Omega_3x+\Omega_2\Lambda_3)^{k+m}}$. After some mathematical manipulations, we can complete the proof of Lemma~\ref{Lem_user_outage}.


\vspace{-0mm}
\section{Proof of Lemma~\ref{Lem_user_outage_asy}}\label{App_lem2}
1) {{\it {\bf Case~1}}}:
For brevity, we set $\Lambda_1 = \Lambda_2 = \Lambda_3=\Lambda$. Thus, $F_{\gamma_{ij}} (\gamma_{th})$ in (\ref{eq_Fgammaij_gx}) can be rewritten as
\begin{equation}
\label{eq_Fgammaij_gx_diversity}
\begin{split}
     F_{\gamma_{ij}}(x) = & \frac{\gamma(m,\frac{mx}{\Omega_1\Lambda})}{\Gamma(m)} +  \frac{\Gamma(m,\frac{mx}{\Omega_1\Lambda})}{\Gamma(m)^2}g(x)\\
     &+\frac{\gamma(m,\frac{m}{\Omega_3})\Gamma(m,\frac{mx}{\Omega_1\Lambda})\gamma(m,\frac{mx}{\Omega_2\Lambda})}{\Gamma(m)^3},
\end{split}
\end{equation}
where
    $g(x) = \left(\frac{m\Lambda}{\Omega_3}\right)^m\int_{0}^{\Lambda} \frac{\gamma(m,\frac{mx}{\Omega_2z})e^{-\frac{m\Lambda}{\Omega_3z}}}{\Gamma(m)z^{m+1}}dz$.
Re-writing the lower incomplete gamma function for $x\rightarrow 0$ as
    $\gamma(m,x) \approx   \frac{x^m}{m}$
and using $\int_0^x \frac{e^{-\frac{k}{t}}}{t^{m+1}}dt = k^{-m} \Gamma(m,\frac{k}{x})$, $g(x)$ can be approximated  as
    $g(x)
    \approx \frac{\Omega_3^m \Gamma(2m,\frac{m}{\Omega_3})}{m \Gamma(m) \Omega_2^m} \left(\frac{x}{\Lambda}\right)^m$.
Applying this approximation to  (\ref{eq_Fgammaij_gx_diversity}) leads to
    $F_{\gamma_{ij}} (x) 
    \approx G(m) \left(\frac{x}{\Lambda}\right)^m + \mathcal{O}\left(\frac{1}{\Lambda^{m+1}}\right)$
where $G(m) = \frac{m^{m-1}}{\Gamma(m)\Omega_1^m} + \frac{m^m \gamma(m,\frac{m}{\Omega_3})+\Omega_3^m\Gamma(2m,\frac{m}{\Omega_3})}{m\Gamma(m)^2 \Omega_2^m}$. Thus, for large $\Lambda\rightarrow\infty$, an asymptotic expression for user outage probability in \eqref{eq_user_outage} is
\begin{equation*}
\label{eq_user_outage_asya_nakagami}
P_{o}^{(u)}  \rightarrow  \hspace{-3mm} \sum_{k = 1}^{(M-1)N +1} \sum_{i=0}^{k-1}  \frac{{\mathbb P}(\gamma^{(k)})(MN)!\binom{k-1}{i}\hspace{-1mm}\left(\hspace{-1mm}\frac{G(m)\gamma_{th}^m}{\Lambda^m} \hspace{-1mm}\right)^{\mu(k,i)}}{(-1)^i\mu(k,i)(k-1)!(MN-k)!},
\end{equation*}
where $\mu(k,i)\hspace{-1mm}=\hspace{-1mm}M\hspace{-1mm}N\hspace{-0.5mm}-\hspace{-0.5mm}k\hspace{-0.5mm}+\hspace{-0.5mm}i\hspace{-0.5mm}+\hspace{-0.5mm}1$.
We can extract the dominant term of $P_{o}^{(u)}$ when $k\hspace{-1mm}=\hspace{-1mm}(M\hspace{-1mm}-\hspace{-1mm}1)N \hspace{-1mm}+\hspace{-1mm}1$ and $i\hspace{-1mm}=\hspace{-1mm}0$, which results in
\begin{equation}
\label{eq_user_outage_asyb_nakagami}
\begin{split}
P_{o}^{(u)} \rightarrow & \frac{{\mathbb P}(\gamma^{((M-1)N +1)})(MN)!}{N ((M-1)N)!(N-1)!} \left(\frac{G(m)\gamma_{th}^m}{\Lambda^m} \right)^{N}.
\end{split}
\end{equation}
 For the worst case scenario, i.e., $\gamma_{(u)}=\gamma^{((M-1)N+1)}$, the corresponding probability can be calculated as
\begin{equation}\label{eq_worst_prob}
    {\mathbb P}\left(\gamma^{((M-1)N+1)}\right) = \left\{\begin{array}{ll}
        \frac{1}{M}\prod_{i=1}^{N-1} \frac{N-i}{MN-i}, & M > N \\
        \frac{2}{M}\prod_{i=1}^{N-1} \frac{N-i}{MN-i}, & M = N.
    \end{array}
    \right.
\end{equation}
Substituting \eqref{eq_worst_prob} into \eqref{eq_user_outage_asyb_nakagami}, we complete the proof of Case~1.

2) {{\it {\bf Case~2} (Outage floor)}}: When $\Lambda_2 = \Lambda
\rightarrow \infty$ with fixed $\Lambda_1$ and $\Lambda_3$, we have the asymptotic expression in \eqref{eq_user_outage_asy2} by noting that $\Gamma(m,\frac{m\Lambda_3}{\Omega_3\Lambda_2})\rightarrow\Gamma(m)$, $\gamma(m,\frac{m\Lambda_3}{\Omega_3\Lambda_2})\rightarrow 0$ and $\Gamma\left(k+m, \frac{m(\Omega_3x+\Omega_2\Lambda_3)}{\Omega_2\Omega_3\Lambda_2} \right)\rightarrow\Gamma(k+m)$ in \eqref{eq_user_outage}.

\vspace{-0mm}
\section{Proof of Lemma~\ref{Lem_user_outage_imCSI}}\label{App_lem3}
Following Appendix~\ref{App_lem1}, from \eqref{eq_Fgammaijy} when $m=1$, the CDF of $\gamma_{ij}$ for a given $\hat Q_{ij} = y$ is given by
\begin{equation}
    F_{\gamma_{ij}|\hat Q_{ij}}(x|y) = 1 - e^{-x\left(\frac{\Omega_{e_1}+d_1^\beta/\Lambda_1}{\Omega_{\hat h_1}}+\frac{\Omega_{e_2}}{\Omega_{\hat h_2}} \right) - \frac{d_2^\beta N_0 x}{\Omega_{\hat h_2}y}},
\end{equation}
where $F_{{\gamma}_{ij,1}}(x)\hspace{-1mm} = \hspace{-1mm} 1-e^{-\frac{x}{\Omega_{\hat h_1}}\left(\Omega_{e_1}+\frac{d_1^\beta}{\Lambda_1} \right)}$, and $F_{{\hat \gamma}_{ij,2}|\hat Q_{ij}}(x|y) = 1-e^{-\frac{x}{\Omega_{\hat h_2}}\left(\Omega_{e_2}+\frac{d_2^\beta N_0}{y} \right)}$.
By noting that
\begin{equation}
    f_{\hat Q_{ij}} (y) \hspace{-1mm} = \hspace{-1mm} \frac{d_3^\beta I_{max}}{\Omega_{\hat f} y^2} e^{\hspace{-1mm} -\frac{d_3^\beta I_{max}}{\Omega_{\hat f} y}}U(y) + \hspace{-1mm} \left( \hspace{-1mm} 1\hspace{-1mm} -\hspace{-1mm} e^{\hspace{-1mm} -\frac{d_3^\beta \Lambda_3}{\Omega_{\hat f}\Lambda_2}}\hspace{-1mm} \right) \hspace{-1mm} D(y),
\end{equation}
the CDF is calculated as $F_{\gamma_{ij}}(x) = \int F_{\gamma_{ij}|\hat Q_{ji}} (x|y) f_{\hat Q_{ji}}(y) dy$. With straightforward mathematical manipulations, we can solve this as in \eqref{eq_Fgammaij_impcsi},
which completes the proof. 


\vspace{-0mm}
\section{Proof of Lemma~\ref{Lem_user_throughput}}\label{App_lem4}
By order statistics, $F_{\gamma_{(u)}}$ can be rewritten as
\begin{equation}\label{eq_user_SNR_cCDF}
    F_{\gamma_{(u)}}\hspace{-0.5mm}(x) \hspace{-1mm}= \hspace{-1mm}1 - \hspace{-4mm} \sum_{k=1}^{(\hspace{-0.5mm}M\hspace{-0.5mm}-\hspace{-0.5mm}1\hspace{-0.5mm})\hspace{-0.5mm}N\hspace{-0.5mm}+\hspace{-0.5mm}1}\sum_{i=0}^{M\hspace{-0.5mm}N\hspace{-0.5mm}-\hspace{-0.5mm}k}\hspace{-1mm}\frac{(M\hspace{-0.5mm}N)!\binom{M\hspace{-0.5mm}N\hspace{-0.5mm}-\hspace{-0.5mm}k}{i}{\mathbb P}(\hspace{-0.5mm}\gamma^{(k)}\hspace{-0.5mm})(1\hspace{-1mm}-\hspace{-1mm}F_{\gamma_{ij}}\hspace{-0.5mm}(\hspace{-0.5mm}x\hspace{-0.5mm})\hspace{-0.5mm})^{k+i}}{(-1)^i (k+i) (k-1)!(MN-k)!},
\end{equation}
where 
\begin{equation}\label{eq_CDF_gammaij_Rayleigh}
    F_{\gamma_{ij}}(x) = 1 - e^{-ax}(b+\frac{c}{x+d})
\end{equation}
by set $m=1$ in \eqref{eq_user_outage} with $a = \frac{1}{\Omega_1\Lambda_1}+\frac{1}{\Omega_2\Lambda_2}$, $b = 1-e^{-\frac{\Lambda_3}{\Omega_3\Lambda_2}}$, $c = d (1-b)$ and $d = \frac{\Omega_2\Lambda_3}{\Omega_3}$. The average throughput of each user can be given by \cite[eq.17]{Atapattu19}
\begin{equation}\label{eq_throughput}
    \bar \tau_{(u)} 
    = \frac{1}{2M\ln 2}\int_0^\infty \frac{1-F_{\gamma_{(u)}}(x)}{1+x} d x
\end{equation} 
From \eqref{eq_user_SNR_cCDF} and \eqref{eq_throughput}, the average throughput of each user is given by
\begin{align}\label{eq_user_throughput_qx}
    \bar \tau_{(u)} \hspace{-1mm}=\hspace{-1mm} \frac{1}{2M\ln 2}\hspace{-3mm}\sum_{k=1}^{(M-1)N+1}\sum_{i=0}^{MN-k}\frac{(-1)^i(MN)!\binom{MN-k}{i}}{ t (k-1)!(MN-k)!}\nonumber\\
    \underbrace{\int_0^\infty \frac{(e^{-ax}(b+\frac{c}{x+d})^t)}{x+1}dx}_{:=q(x)},
\end{align}
where $t = k+i$. Applying binomial expansion to $(b+\frac{c}{x+d})^t$, $q(x)$ can be given by
\begin{equation}\label{eq_qx}
    q(x) = \sum_{j=0}^{t}\binom{t}{j}b^{t-j}c^j \underbrace{\int_0^\infty \frac{e^{-atx}}{(x+1)(x+d)^j}dx.}_{:=h(j)}
\end{equation}
If $d = 1$, $h(j)$ is given by
\begin{equation}\label{eq_hj_d=1}
    h(j|d=1) = \int_0^\infty\frac{e^{-atx}}{(x+1)^{j+1}}dx.
\end{equation}
If $d\neq 1$, then by partial fraction, 
\begin{equation}\label{hj_partial_fraction}
    \frac{1}{(x+1)(x+d)^j} = \frac{(d-1)^{-j}}{x+1} - \sum_{r=1}^{j}\frac{(d-1)^{r-j-1}}{(x+d)^r}.
\end{equation}
Therefore, $h(j)$ is given by
\begin{equation}\label{eq_hj_pf}
    h(j|d\neq 1) = \frac{\int_0^\infty\frac{e^{-atx}}{x+1}dx}{(d-1)^{j}} - \sum_{r=1}^{j}\frac{\int_0^\infty\frac{e^{-atx}}{(x+d)^r}dx}{(d-1)^{j-r+1}}.
\end{equation}
Applying \cite[Eqs.~3.352.4 and 3.353.2]{TISP07} to \eqref{eq_hj_d=1} and \eqref{eq_hj_pf}, we can derive $h(j)$ in \eqref{eq_hj_1} and \eqref{eq_hj_2} and then complete the proof.

\end{appendices}
\vspace{-0mm}
\bibliographystyle{IEEEtran}
\bibliography{0reference,IEEEabrv}

\end{document}